\documentclass{sig-alternate-2013}
\permission{Copyright is held by the International World Wide Web Conference Committee (IW3C2). IW3C2 reserves the right to provide a hyperlink to the author's site if the Material is used in electronic media.}
\conferenceinfo{WWW 2015,}{May 18--22, 2015, Florence, Italy.}
\copyrightetc{ACM \the\acmcopyr}
\crdata{978-1-4503-3469-3/15/05. \\
http://dx.doi.org/10.1145/2736277.2741642}

\usepackage{bold-extra}
\usepackage{amssymb}
\usepackage{amsmath}

\usepackage{amsthm}
\usepackage{graphicx}
\usepackage{comment}
\usepackage{url}
\usepackage{todonotes}
\usepackage{hyperref}
\usepackage{listings}
\usepackage{todonotes}
\usepackage{booktabs}
\usepackage{amsmath,amssymb,amsthm}
\usepackage{algorithmic}
\usepackage{algorithm}
\usepackage{tikz}
\usepackage{tabularx}
\usepackage{subfigure}
\usepackage{wrapfig}
\usetikzlibrary{shapes.gates.logic.US,trees,positioning,arrows}
\newtheorem{thm}{Theorem}
\newtheorem{definition}{Definition}
\DeclareRobustCommand{\hr3}{{$\mathcal{HR}^3$}}

\usepackage{xspace}

\begin{document}

\title{ROCKER -- A Refinement Operator for Key Discovery} %
\numberofauthors{3}

\author{
\alignauthor
Tommaso Soru\\ 
       \affaddr{Institute of Computer Science,\\University of Leipzig}\\
       \email{tsoru@informatik.uni-leipzig.de}
\alignauthor
Edgard Marx\\ 
       \affaddr{Institute of Computer Science,\\University of Leipzig}\\
       \email{marx@informatik.uni-leipzig.de}
\alignauthor
\mbox{Axel-Cyrille Ngonga Ngomo}\\
       \affaddr{Institute of Computer Science,\\University of Leipzig}\\
       \email{ngonga@informatik.uni-leipzig.de}
}

\maketitle

\begin{abstract}
The Linked Data principles provide a decentral approach for publishing structured data in the RDF format on the Web. 
In contrast to structured data published in relational databases where a key is often provided explicitly, finding a set of properties that allows identifying a resource uniquely is a non-trivial task. 
Still, finding keys is of central importance for manifold applications such as resource deduplication, link discovery, logical data compression and data integration. 
In this paper, we address this research gap by specifying a refinement operator, dubbed ROCKER, which we prove to be finite, proper and non-redundant. 
We combine the theoretical characteristics of this operator with two monotonicities of keys to obtain a time-efficient approach for detecting keys, i.e., sets of properties that describe resources uniquely. 
We then utilize a hash index to compute the discriminability score efficiently. 
Therewith, we ensure that our approach can scale to very large knowledge bases. 
Results show that ROCKER yields more accurate results, has a comparable runtime, and consumes less memory w.r.t. existing state-of-the-art techniques.
\end{abstract}




\section{Introduction}
The number of facts published in the Linked Data Web has grown considerably over the last years~\cite{AUE+11}.
In particular, large knowledge bases such as LinkedTCGA~\cite{saleem2013linked} and LinkedGeoData~\cite{stadler2012linkedgeodata} encompass more than 20 billion triples each. 
The architectural principles behind the Linked Data Web are akin to those on the Web.
In particular, the decentral data publication process leads to facts on the same real-world entities being published across manifold knowledge bases.
For example, information on \textit{Austin, Texas} is distributed across several knowledge bases, including DBpedia\footnote{\url{http://dbpedia.org}}, LinkedGeoData and GeoNames\footnote{\url{http://www.geonames.org/}}.
Given the size of the current Linked Data datasets, providing unique means to characterize resources within existing datasets would facilitate the use of these knowledge bases, for example within the context of entity search, data integration, linked data compression and link discovery~\cite{pernelle2013automatic}.
Especially for the link discovery task, being provided with unique descriptions of resources in a knowledge base would allow for the more time-efficient computation of property matchings for link specifications, a task that has been pointed out to be particularly tedious in previous work~\cite{cheatham2013string}.
\tikzset{global scale/.style={
    scale=#1,
    every node/.style={scale=#1}
  }
}

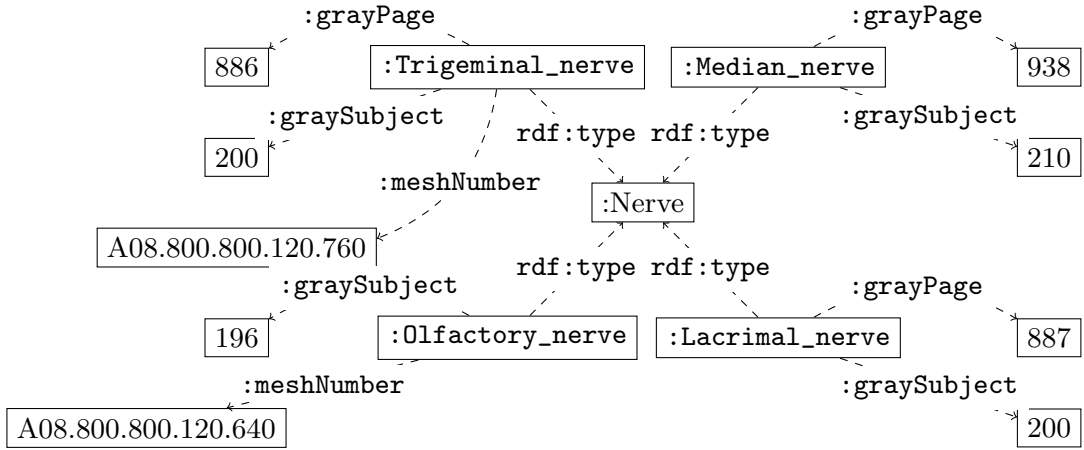
\begin{figure*}[htb]
\centering
\begin{tikzpicture}[scale=1.2, every node/.style={scale=1.2}]
\node (class) at (5.5,3.5) [rectangle,draw]{:Nerve};

\node (jd1)  at (4,5) [rectangle,draw]{\texttt{:Trigeminal\_nerve}};
\node (jd1label) at (1,5) [rectangle,draw]{886};
\node (jd1Ssn) at (1,4) [rectangle,draw]{200};
\node (jd1Id) at (1,3) [rectangle,draw]{A08.800.800.120.760};

\draw[->, dashed, bend right] (jd1) to node[fill=white]{\texttt{:grayPage}} node[auto, swap] { } (jd1label);
\draw[->, dashed] (jd1) to node[fill=white]{\texttt{:graySubject}} node[auto, swap] { } (jd1Ssn);
\draw[->, dashed, bend left] (jd1) to node[fill=white]{\texttt{:meshNumber}} node[auto, swap] { } (jd1Id); 
\draw[->, dashed] (jd1) to node[fill=white]{\texttt{rdf:type}} node[auto, swap] { } (class); 

\node (jd2) at (7,5) [rectangle,draw]{\texttt{:Median\_nerve}};
\node (jd2label) at (10,5) [rectangle,draw]{938};
\node (jd2Ssn) at (10,4) [rectangle,draw]{210};

\draw[->, dashed, bend left] (jd2) to node[fill=white]{\texttt{:grayPage}} node[auto, swap] { } (jd2label);
\draw[->, dashed] (jd2) to node[fill=white]{\texttt{:graySubject}} node[auto, swap] { } (jd2Ssn);
\draw[->, dashed] (jd2) to node[fill=white]{\texttt{rdf:type}} node[auto, swap] { } (class); 

\node (jd3) at (7,2) [rectangle,draw]{\texttt{:Lacrimal\_nerve}};
\node (jd3label) at (10,2) [rectangle,draw]{887};
\node (jd3Age) at (10,1) [rectangle,draw]{200};

\draw[->, dashed, bend left] (jd3) to node[fill=white]{\texttt{:grayPage}} node[auto, swap] { } (jd3label);
\draw[->, dashed] (jd3) to node[fill=white]{\texttt{:graySubject}} node[auto, swap] { } (jd3Age);
\draw[->, dashed] (jd3) to node[fill=white]{\texttt{rdf:type}} node[auto, swap] { } (class); 

\node (jd4)  at (4,2) [rectangle,draw]{\texttt{:Olfactory\_nerve}};
\node (jd4label) at (1,2) [rectangle,draw]{196};
\node (jd4Category) at (0,1) [rectangle,draw]{A08.800.800.120.640};

\draw[->, dashed, bend right] (jd4) to node[fill=white]{\texttt{:graySubject}} node[auto, swap] { } (jd4label);
\draw[->, dashed] (jd4) to node[fill=white]{\texttt{:meshNumber}} node[auto, swap] { } (jd4Category);
\draw[->, dashed] (jd4) to node[fill=white]{\texttt{rdf:type}} node[auto, swap] { } (class);

%
%

\end{tikzpicture}
\caption{Fragment from a knowledge base on human nerves. The fragment was extracted from DBpedia 3.9.}
\label{fig:example}
\end{figure*}

In relational databases, keys are commonly either artificial or sets of columns that allow to describe each resource uniquely. 
Previous works~\cite{pernelle2013automatic,atencia2014b,symeonidou2014} adopt this approach for uniquely describing RDF data and use properties instead of columns.
Several problems occur when trying to detect keys for RDF data.
\begin{enumerate}
\item Resources from the same datasets might not all have the same properties.
For example, in the fragment of DBpedia 3.9 shown in Figure~\ref{fig:example}, only 50\% of the resources have a \texttt{:meshNumber}.
Thus, while the \texttt{:meshNumber} is unique, it cannot be used as a key for this dataset.
\item The inverse problem exists for the \texttt{:graySubject}, which covers all resources but is not unique as the trigeminal nerve and the lacrimal nerve have the same \texttt{:graySubject}.
For our toy dataset, only keys of size larger that 1 exist (e.g., $\{$\texttt{:graySubject}, \texttt{:grayPage}$\}$).
\item The key discovery problem is exponential in the number of properties $n$ in the knowledge base, as the solution space contains $2^n - 1$ possible sets of keys. Thus, na\"ive solutions to the key discovery problem do not scale.
\end{enumerate}
Moreover, depending on the use case, key discovery approaches have to be able to detect a single key (e.g., to link resources within a knowledge base) or to detect all keys for a resource (e.g., when integrating data across knowledge bases).

In this paper, we address the three problems of key discovery within both settings of key discovery (i.e., finding all keys or almost-keys within a given threshold) by using a refinement operator dubbed $\rho$. 
This operator is able to detect sets of properties that describe any instance of a given class in a unique manner.
By these means, it can generate n-tuples of property values that can be used as keys for resources which instantiate a given class.
Our operator relies on a scoring function to compare sets of properties.
Based on this comparison, it can efficiently detect single keys, all keys and even predict whether a key can be found in a given dataset.
In addition to being finite, non-redundant and proper, our operator also scales well and can thus be used on large knowledge bases.
Our contributions are:
\begin{itemize}
\item We provide the first refinement operator for key discovery on RDF knowledge bases.
\item We prove that our operator is finite, non-redundant, proper, but not complete.
\item We utilize the combination of a hash index to compute the discriminability score, i.e. the ability for a set of properties to distinguish their subjects, with two monotonicities of keys to prune the refinement tree and thus ensure that our operator scales.
\item We show that our approach succeeds on datasets where current state-of-the-art approaches fail.
\item We evaluate our operator on the OAEI instance matching benchmark datasets as well as on DBpedia classes with large populations. In particular, we measure the overall runtime, the memory consumption and the reduction ratio of our approach. Our results suggest that we outperform the state of the art w.r.t. correctness and memory consumption. Moreover, our results suggest that our approach terminates within an acceptable time frame even on very large datasets.
\end{itemize}

The rest of this paper is structured as follows:
We begin by defining the problem at hand formally.
Thereafter, we present our operator and prove its theoretical characteristics.
After a discussion of related work, we evaluate our operator on synthetic and real data.
We then conclude and present some future work.

\section{Preliminaries}
In the following, we formalize the definition of keys that underlie this paper. This definition is used by our refinement operator to efficiently detect keys.

\subsection{Keys}
Let $K$ be a finite RDF knowledge base containing instances which belong to a given class and their Concise Bounded Description (CBD).\footnote{For the definition of CBD, see \url{http://www.w3.org/Submission/CBD/}.}
$K$ can be regarded as a set of triples $(s, p, o) \in (\mathcal{R}  \cup \mathcal{B}) \times \mathcal{P} \times (\mathcal{R} \cup \mathcal{L} \cup \mathcal{B})$, where $\mathcal{R}$ is the set of all resources, $\mathcal{B}$ is the set of all blank nodes, $\mathcal{P}$ the set of all predicates and $\mathcal{L}$ the set of all literals.
We call two resources $r_1, r_2 \in \mathcal{R}$ distinguishable w.r.t. a set of properties $P=\{p_1, \ldots p_n\}$ iff $\exists p \in \{p_1, \ldots p_n\}: \neg ((r_1, p, o) \wedge (r_2, p, o))$.
Given a knowledge base $K$, the idea behind \emph{key discovery} is to find one or all sets of properties which make their respective subjects distinguishable in $K$.

\begin{definition}[Key]
We call a set of properties $P \subseteq \mathcal{P}$ a \emph{key} for a knowledge base $K$ (short: key, denoted $key(P, K)$) if all resources in $K$ are distinguishable w.r.t. $P$.
\end{definition}

\begin{definition}[Minimal key]
We call $P$ a \emph{minimal key} (short: mkey) iff $P$ is a key but none of its subsets is. Formally, 
\begin{equation}
mkey(P, K) \Rightarrow key(P, K) \wedge (\neg \exists P' \subset P: key(P',K)).
\end{equation}
\end{definition}

\subsection{Discriminability} \label{sec:discr}
A \emph{key} for an RDF knowledge base and a \emph{primary composite key} for a database share the same aim.
RDF properties represent the projection of database fields into the RDF paradigm, as well as each resource represents a tuple.
However, while a tuple element has only one single value, a property might link a resource to more than one RDF objects.
Therefore, two resources are distinguishable from each other w.r.t. a set of properties $P$ if their sets of objects are different for at least one $p \in P$.

To the best of our knowledge, this particular feature of keys was not taken into account by previous works on key discovery for RDF data~\cite{pernelle2013automatic,symeonidou2014,atencia2014b}.
For instance, \cite{pernelle2013automatic,symeonidou2014} consider two resources $r$ and $r'$ as not distinguishable w.r.t. $P$ if for each $p \in P$ they share at least one object.

\begin{figure*}[htb]
	\centering
	\includegraphics[width=1.4\columnwidth]{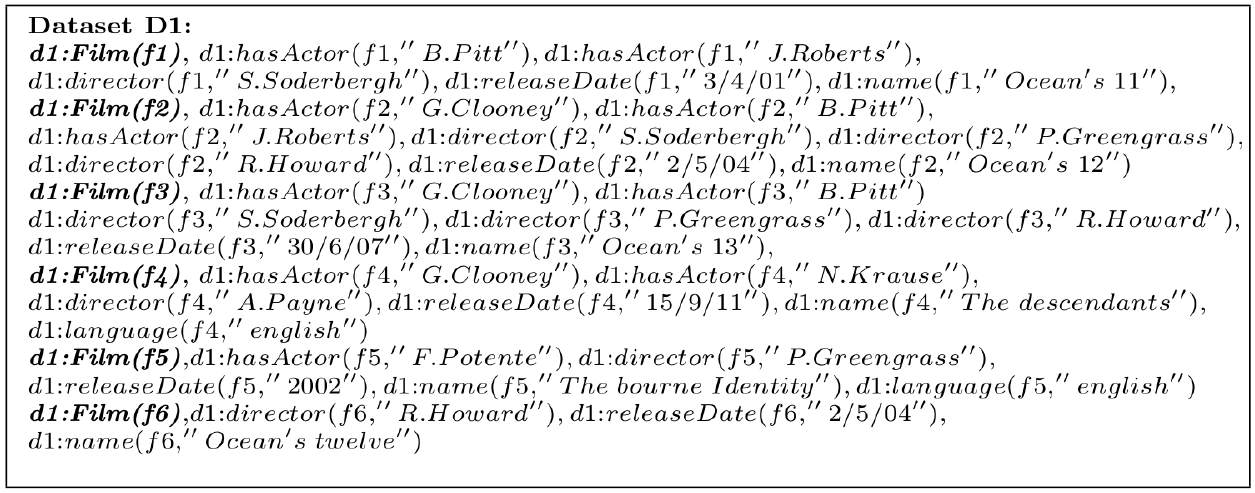}
	\caption{Example of RDF data, as reported in Symeonidou et al., 2014.}
	\label{fig:rdfdata}
\end{figure*}

Figure~\ref{fig:rdfdata} shows an example of RDF data, as reported in 
\cite{symeonidou2014}.
Here, the authors claim that $P=\{p_1\}=\{$\texttt{:hasActor}$\}$ is not a key because \texttt{"G.Clooney"} is the object of more than one instance of \texttt{:Film}.
We would instead consider $P$ as a key, since every film is linked with a different set of objects ($sobj$), i.e.:
\begin{flushleft} \footnotesize
$sobj( \texttt{:f1},p_1)=\{ \texttt{"B.Pitt"}, \texttt{"J.Roberts"} \} $\\
$sobj( \texttt{:f2},p_1)=\{ \texttt{"G.Clooney"}, \texttt{"B.Pitt"}, \texttt{"J.Roberts"} \} $\\
$sobj( \texttt{:f3},p_1)=\{ \texttt{"B.Pitt"}, \texttt{"G.Clooney"} \} $\\
$sobj( \texttt{:f4},p_1)=\{ \texttt{"G.Clooney"}, \texttt{"N.Krause"} \} $\\
$sobj( \texttt{:f5},p_1)=\{ \texttt{"F.Potente"} \} $\\
$sobj( \texttt{:f6},p_1)=\varnothing$
\end{flushleft}

Note that \texttt{:f6} is still distinguishable from the other resources w.r.t. $P$, since no other instance of \texttt{:Film} in the knowledge base has 0 actors.
This particular case was not considered, for example, by the authors of \cite{atencia2014b}.

\subsection{Properties of a key}
Keys abide by several monotonicities~\cite{pernelle2013automatic}. 
The first is the so-called \emph{key monotonicity}, which is given by 
\begin{equation}
key(P, K) \Rightarrow \forall P': P \subseteq P' \Rightarrow key(P', K).
\end{equation}
The reciprocal monotonicity is called the \emph{non-key monotonicity}, which is given by 
\begin{equation}
\neg key(P, K)  \Rightarrow (\forall P' \subseteq P: \neg key(P', K)). 
\end{equation}
In other words, adding a property to a key yields another key, whilst removing a property to a non-key yields another non-key.
In this paper, we present a key discovery approach based on refinement operators.

\section{A Refinement Operator for Key Discovery}
In this section, we present our refinement operator for key discovery and prove some of its theoretical characteristics. Our formalization is based on that presented in~\cite{DBLP:journals/ml/LehmannH10}. 

Let $P \subseteq \mathcal{P}$. 
Moreover, let $score: 2^\mathcal{P} \rightarrow [0, 1]$ be a function that maps each subset $P$ of $\mathcal{P}$ to the fraction of subject resources from $K$ that are distinguishable by using $P$.
\begin{thm}[Induced quasi-ordering] 
The score function induces a quasi-ordering $\preceq$ over the set $\mathcal{P}$, which we define as follows:
\begin{equation} 
P_1 \preceq P_2 \Leftrightarrow \min\limits_{p \in P_1} score(p) \leq \min\limits_{q \in P_2} score(q).
\end{equation} 
\end{thm}  
The reflexivity and transitivity of $\preceq$ are direct consequences of the reflexivity and transitivity of $\leq$ in $\mathbb{R}$.
Note that $\preceq$ is not antisymmetric as two sets of properties $P_1$ and $P_2$ can be different and contain the property with the lowest score, leading to $P_1 \preceq P_2$ and $P_2 \preceq P_1$.
\begin{definition}[Refinement Operator] 
Given a quasi-ordered space $(S, op)$ an upward refinement operator $r$ is a mapping from $S$ to $2^S$ such that $\forall s \in S: s' \in r(s) \Rightarrow op(s, s')$. $s'$ is then called a generalization of $s$.
\end{definition} 
We define our refinement operator over the space $(\mathcal{P}, \preceq)$.
First, we begin by ordering the elements of $\mathcal{P}$ according to their score in ascending order, i.e., $\forall p_i, p_j \in \mathcal{P}, i \leq j \Rightarrow score(p_i) \leq score(p_j)$.
Then, we can define our operator as follows:
\begin{equation}
\rho(P) = 
\begin{cases}
\mathcal{P} \mbox{ iff } P = \emptyset, \\ 
\{P \cup \{p_1\}, \ldots, P \cup \{p_i\}\} \mbox{ where } p_j \in P \Rightarrow i < j.
\end{cases}
\end{equation}
For example, the complete refinement graph for $\mathcal{P} = \{p_1, p_2, p_3\}$ is given in Figure \ref{fig:refinement}.
We use this operator in an iterative manner by only expanding the node with the highest score in the refinement graph.
The intuition behind this approach to searching for key is that by ordering properties by their score, we can easily detect and expand the most promising sets of properties without generating redundant nodes. 
To prove some of the characteristics of $\rho$, we need to explicate the concept of a refinement chain:

\begin{definition}[Refinement chain]
A set $P_2 \in \mathcal{P}$ belong to the refinement chain of $P_1 \in \mathcal{P}$ iff $\exists k \in \mathbb{N}: P_2 \in \rho^k(P_1), $ where $\rho^{k}(P) = \begin{cases}
\mathcal{P} \mbox{ iff } k = 0, \\ 
\rho(\rho^{k-1}(P)) \mbox{ else }.
\end{cases}$
\end{definition}
For example, a refinement chain exists between $\{p_3\}$ and $\{p_1, p_2, p_3\}$ in the example shown in Figure~\ref{fig:refinement}. There is yet no refinement chain between $\{p_1\}$ and $\{p_2\}$ in the same example.

\begin{figure}
\centering
\begin{tikzpicture}
\centering
\node[draw] (bottom) at (0,0) {$\emptyset$} ;
\node[draw] (p1) at (-2.5,1) {$\{p_1\}$}  edge [<-] (bottom);
\node[draw] (p2) at (0, 1) {$\{p_2\}$}  edge [<-] (bottom);
\node[draw] (p3) at (2.5,1) {$\{p_3\}$} edge [<-] (bottom);
\node[draw] (p2p1) at (0,2) {$\{p_1, p_2\}$} edge [<-] (p2);
\node[draw] (p1p3) at (1.5,2) {$\{p_1, p_3\}$}  edge [<-] (p3);
\node[draw] (p2p3) at (3.5,2) {$\{p_2, p_3\}$} edge [<-] (p3);
\node[draw] (p1p2p3) at (3.5,3) {$\{p_1,p_2, p_3\}$} edge [<-] (p2p3);
\end{tikzpicture}	
\caption{Complete refinement graph for $\mathcal{P} = \{p_1, p_2, p_3\}$. The nodes of the graph are subsets of $\mathcal{P}$. A directed edge $(a,b)$ means $b \in \rho(a)$.}
\label{fig:refinement}
\end{figure}
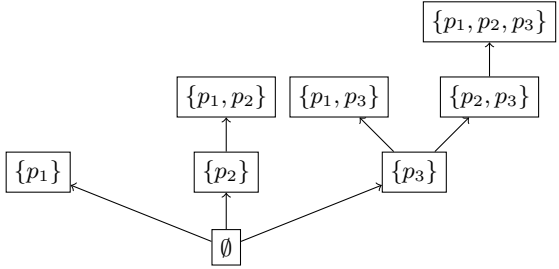

A refinement operator $r$ over the quasi-ordered space $(S, op)$ can abide by the following criteria.
\begin{definition}[Finiteness]
$r$ is finite iff $r(s)$ is finite for all $s \in S$.
\end{definition}
\begin{definition}[Properness]
$r$ is proper if $\forall s \in S, s' \in r(s) \Rightarrow s \neq s'$. 
\end{definition}
\begin{definition}[Completeness]
$r$ is said to be complete if for all $s$ and $s'$, $op(s', s)$ implies that there is a refinement chain between $s$ and $s'$.
\end{definition}
\begin{definition}[Redundancy]
A refinement operator $r$ over the space $(S, op)$ is redundant if two different refinement chains can exist between $s \in S$ and $s'\in S$.
\end{definition}

In the following, we show that $\rho$ is finite, proper and non-redundant but not complete.
\begin{thm}
$\rho$ is finite when applied to a finite knowledge base $K$.
\end{thm}
\begin{proof}
The finiteness of $\rho$ is a direct result of $K$ being finite. The upper bound of the number of properties in $K$ is the number of triples in $K$. Thus, $|K| < \infty \Rightarrow |\mathcal{P}| < \infty$. Per definition, $|\rho(P)| \leq |\mathcal{P}|$. Thus, we can conclude that $\forall P \in \mathcal{P}: |\rho(P)| < \infty$.
\end{proof}

\begin{thm}
$\rho$ is proper.
\end{thm}
\begin{proof}
The properness of $\rho$ also results from the definition of $\rho$. As we always add exactly a property to $P$ when computing $\rho(P)$, we know that $|\rho(P)| = |P| + 1$. Thus, $\rho(P) \neq P$ must hold. 
\end{proof}

\begin{thm}
$\rho$ is not complete.
\end{thm}
\begin{proof}
The incompleteness of $\rho$ follows from the definition of $\rho(\emptyset)$. Let $\mathcal{P} = \{p_1, \ldots, p_n\}$. Then $\{p_1\} \preceq \{p_n\}$. Yet, there is clearly no refinement chain between $\{p_n\}$ and $\{p_1\}$ as any subset of $\mathcal{P}$ connected to $p_n$ via a refinement chain must have a magnitude larger than one. Yet, the magnitude of $\{p_1\}$ is 1, which shows that $\{p_1\}$ cannot be connected to $\{p_n\}$  via a refinement chain.
\end{proof}

\begin{thm}
$\rho$ is not redundant.
\end{thm}
\begin{proof}
$\rho$ being redundant would mean that a pair of property sets $(P, P')$ exist, where $P$ is linked to $P'$ by two distinct refinement chains $C_1$ and $C_2$.
Given that these two chains begin and end at the same node, there must be a node $N$ that is common to the two chains but has two distinct fathers $N_1$ and $N_2$ that are such that $N_1$ belongs to $C_1$ and not to $C_2$ while $N_2$ belong to $C_2$ but not to $C_1$.
Now, $N_1$ being the father of $N$ means $\exists p_i \in \mathcal{P}: N = N_1 \cup {p_i}$.
Conversely, $N_2$ being the father of $N$ also means that $\exists p_j \in \mathcal{P}: N = N_2 \cup {p_j}$.
Now if $N_1 \neq N_2$, then we can assume wlog that $i < j$.
For $N \in \rho(N_2)$ to hold, $j$ resp. $i$ must be less than the index of any element of $N_2$ resp. $N_1$.
Moreover, for $N_1 \cup {p_i} = N_2 \cup {p_j}$ to hold, $p_i$ would have to already be a element of $N_2$.
However, by virtue of the construction of $\rho$, this means that $(N_2 \cup \{p_j\}) \notin \rho(N_2)$ given that $p_i \in N_2$ and $i < j$.
Thus, we can conclude that there cannot be any to distinct refinement pairs between two subsets of $\mathcal{P}$.
\end{proof}

\section{Approach}
In this section, we present \textbf{ROCKER}, a \underline{r}efinement \underline{o}perator approa\underline{c}h for \underline{ke}y discove\underline{r}y.
Our approach was designed with scalability in mind. 
To this end, it implements a scalable version of the discriminability $score$ function based on a hash index.
Moreover, we use the monotonicities of keys to check for the existence of keys as well as decide on nodes that need not be refined.

\subsection{Implementation}
In order to increase the scalability of our operator, we chose an hybrid approach using both in-memory and disk storage database.
The following tasks are then performed by ROCKER:
\begin{enumerate}
\item the knowledge base model is built using the Apache Jena library;
\item for each class, its instances and properties are extracted;
\item this information is stored and indexed over a \emph{B-tree} structure, whereas the instance URI is used as a key;
\item object values are sorted alphabetically, imploded into a string, indexed using hash codes, and stored into each tuple element;
\item the refinement operator starts from the empty-set node;
\item at each node, the discriminability score is computed by performing a query to the database;
\item the computation terminates according to some rule.
\end{enumerate}

\subsection{Definition of the score function} \label{sec:score}
We firstly define the set of subjects of the knowledge base, i.e. the instances of a given class.
\begin{equation}
S = \left\{ s: \exists (s, p, o) \in K \right\}
\end{equation}
We then provide the definition of discriminability for two resources w.r.t. a set of properties $P$.
\begin{flalign}
discr(s, s', P) \Leftrightarrow \\ \forall s \in S \ \forall p \in P \ \nexists s' \in S : sobj(s,p) \equiv sobj(s',p)
\end{flalign}
where $sobj$ is the ``set of objects'' function introduced in Section~\ref{sec:discr}.
Finally, we define the score of $P$ (denoted $score(P)$) as the number of subject resources of $K$ that are distinguishable w.r.t. $P$ by using the following formula:
\begin{equation} \label{def:score}
score(P) = \frac{ \left\vert \left\{ s \in S: \forall s' \in S \ s \neq s' \Rightarrow discr(s, s', P) \right\} \right\vert }{ \left\vert S \right\vert }.
\end{equation}
The set $P$ is a key if $score(P) = 1$, i.e., if $P$ covers all subject resources from $K$ and all resources are distinguishable w.r.t. $P$.

Basing the refinement on this scoring function has the advantage of allowing ROCKER to cover not only keys but also k-almost-keys \cite{symeonidou2014}, which are defined as follows: $P$ is a k-almost-key if $\exists X \subseteq S: |X| \leq k \land \forall s, s' \in S \backslash X: discr(s, s', P)$. Consider for example the data shown in Figure \ref{fig:rdfdata}. If \texttt{:f2} did not have \texttt{"J. Roberts"} in the list of its actors, then it would not be distinguishable from \texttt{:f3}. In this case, the set \{\texttt{hasActor}\} would be 2-almost-key.
We can derive a minimal score $\alpha$ for a k-almost-key by simply using the maximal magnitude of $X$ within $score(P)$:
\begin{equation} \label{def:score}
|X| \leq k \rightarrow score(P) \geq \frac{|S| - k}{|S|} = \alpha.
\end{equation}
Note that a key is a 0-almost-key. Moreover, every $k$-almost-key with $k \geq 0$ is also a $(k+1)$-almost-key. 

In our implementation, the score function for $P$ is thus calculated by querying the class table for the columns associated with the properties in $P$.
For each row, the returned values are then concatenated and added to a sorted set, where duplicates are discarded automatically by virtue of the definition of set.
The size of this final set represents the numerator for Equation~\ref{def:score}.

\subsection{Refinement Operator}
The pseudo code of ROCKER's refinement operator is shown in Algorithm~\ref{alg:approach}.
Given a set of triples $K$ and a set of properties $\mathcal{P}$, we begin by checking whether our $\rho$-based approach is able to find a key at all.
This can be done by computing $score(\mathcal{P})$. If the score is less than 1 (i.e., if $\mathcal{P}$ is not a key), then we know no key exists by virtue of the non-key monotonicity.
We can thus terminate and return $\varnothing$, unless a threshold $\alpha < 1$ is given.
Under this setting, we terminate if $score(\mathcal{P}) < \alpha$.
Now if $\mathcal{P}$ is a key, then some of its subsets might be minimal keys.
We then begin by sorting the elements of $\mathcal{P}$ w.r.t. their score.
This heuristic tries to make the refinement operator discover keys earlier, so that their descendants can be pruned from the refinement tree, thus decreasing the number of score calculations.
A maximal-priority queue is then initialized, where the priority of an element is its score.
The queue is initialized with the empty set with a priority of 0.
We then take the element $P$ of the queue with the highest priority iteratively and remove it from the queue.
Thank to the non-redundancy of $\rho$, there is no need to check whether $P$ has been seen before.
$P$ is refined to $P'$, whose elements $p$ are then checked iteratively.
We thus evaluate the scores of all elements of $P'$. 
If their score is less than 1, then they are added to the queue.
If their score is 1, then we add them to the solution and do not add them to the queue, as they do not need to be refined any further by virtue of the key monotonicity.
We then return the set of all keys.

Our approach has several advantages due to the theoretical characteristics of $\rho$ and the key monotonicities: 
\begin{enumerate}
\item It terminates quickly if there is no key by virtue of using the non-key monotonicity.
\item It is guaranteed to find all existing minkeys by virtue of the key monotonicity.
\item Using a sorted queue, it encourages node pruning by evaluating the most promising nodes first.
\item It never visits the same node twice due to the non-redundancy of $\rho$.
\item It is ensured to find all existing keys.
\end{enumerate}

\begin{algorithm}[htb]
\begin{algorithmic}[1]
\REQUIRE Set of triples $K$
\STATE $\mathcal{P} = \{p: \exists s, p \mbox{ with } (s, p, o) \in K\}$
\IF{$score(\mathcal{P}) < 1$} 
	\RETURN $\varnothing$;
\ENDIF
\STATE $\mathcal{P} =$ sortByScore$(\mathcal{P})$;
\STATE MaxPriorityQueue q = new Queue();
\STATE Set solution = new Set();
\STATE q.add($\varnothing$, 0); // add $\varnothing$ with priority 0
\WHILE{$\neg$q.isEmpty()}
	\STATE $P' = q.getFirst()$;
	\STATE $q.removeFirst()$;
	\STATE $P = \rho(P')$;
	\FORALL{$p \in P$}
		\STATE $\sigma = score(p)$;
		\IF{$\sigma == 1$}
			\STATE solution.add(p);
		\ELSE
			\STATE q.add(p, $\sigma$);
		\ENDIF
	\ENDFOR
\ENDWHILE
\RETURN solution;
\end{algorithmic}
\caption{ROCKER's algorithm for detecting all keys. The algorithm for detecting a single key does not require the solution variable. Instead, it returns the first $P$ having $score(P)=1$ it finds.} 
\label{alg:approach}
\end{algorithm}

\subsection{Search Strategy}

As already mentioned in \cite{pernelle2013automatic}, the number of nodes to visit in the key discovery problem is exponential w.r.t. the number of properties considered.
More precisely, given $n$ properties, the computational complexity of our algorithm is $O(2^n)$ in the worst case, i.e. when there exists one only key formed by all properties.
We tackle this issue by introducing a \textit{fast search} strategy, which can be enabled to speed up the computation.
Within this optional setting, whenever a key is found, at the next iteration all branches containing parts of the key are pruned from the refinement tree.
This strategy tries to improve the runtime while fostering diversity among the discovered keys.
Moreover, we consider properties whose atomic candidate keys have a score greater than a threshold $\tau$.
This lets the algorithm discard properties that alone distinguish less instances, thus having a lower probability to be part of a key.

\section{Related Work}

Key discovery is a rather new research field within the domain of Linked Data, although the issue of finding keys among fields has been inherited from relational databases.
However, relational databases do not consider semantics (e.g., subsumption relations) which belong to the core of Linked Data.
Previous work on key discovery for the Semantic Web can be found in~\cite{pernelle2013automatic,atencia2014b,symeonidou2014}.
For instance, KD2R is an automatic discovery tool for composite keys in RDF data sources that may conform to different schemata~\cite{pernelle2013automatic}.
It relies on the creation of prefix trees, which serve for finding maximal undetermined keys and non-keys.
However, state-of-the-art approaches as Linkkey and SAKey have shown to outperform KD2R on runtime and number of generated keys~\cite{atencia2014b,symeonidou2014}.
To the best of our knowledge, not only is ROCKER the first refinement-operator-based approach for key discovery, it is also the first machine-learning-based approach for key discovery.

Independently on the application domain, the key discovery problem is a sub-problem of Functional Dependencies (FDs), as every element is distinguishable only by its attributes.
Keys or FDs are widely used in ontology alignment, as well as in data mining~\cite{journals/datamine/MannilaT97}, reverse engineering~\cite{chiang1994reverse}, and query optimization~\cite{mannila1994algorithms,ilyas2004cords}.
In particular, blocking methods such as \cite{conf/aaai/MichelsonK06} utilize approximate keys to reduce the computational complexity of dataset joins.
Unsupervised learning approaches aim at finding links among datasets by comparing datatype values of properties contained into minimal keys~\cite{song2011automatically}.
The so-called collective or global approaches of data linking use keys to generate identity links between instance joins for the final scope of enriching the ontology with the collected information~\cite{sais2009combining,arasu2009large}.

As previously mentioned, one of the main application areas of ROCKER is link discovery.
Several approaches have been developed in previous works to detect matching properties and using them for link discovery.
For example,~\cite{DBLP:conf/semweb/NgomoLAH11} relies on the hospital-residents problem to detect property matches.
Other approaches based on genetic programming (e.g.,~\cite{DBLP:conf/esws/NikolovdM12}) detect matching properties while learning link specifications, which currently implements several time-efficient approaches for link discovery. 
\cite{DBLP:conf/ijcai/NgomoA11} proposes an approach based on the Cauchy-Schwarz inequality that allows discarding a large number of superfluous comparisons.
HYPPO~\cite{DBLP:conf/semweb/Ngomo11} and \hr3~\cite{DBLP:conf/semweb/Ngomo12} rely on space tiling in spaces with measures that can be split into independent measures across the dimensions of the problem at hand.
In particular, \hr3 was shown to be the first approach that can achieve a relative reduction ratio $r'$ less or equal to any given relative reduction ratio $r > 1$.
In the ACIDS approach, similarity measures are performed on property values in order to yield features for machine-learning classifiers as support vectors machines~\cite{SONG12}.
Amongst other link discovery approaches, RDF-AI~\cite{SCH+09} relies on a five-step method that comprises the preprocessing, matching, fusion, interlink and post-processing of data sets.

\begin{table}[ht]\scriptsize
{\caption{Key extraction quality results.}\label{tab:keyquality}}
\begin{tabularx}{\columnwidth}{Xrrr}
\toprule
\textbf{Dataset} & \textbf{ROCKER} & \textbf{Linkkey} & \textbf{SAKey} \\
\midrule
Restaurant 1 & 3 (100\%, 100\%) & 0 (0\%, 0\%) & 8 (62\%, 25\%) \\
Restaurant 2 & 3 (100\%, 100\%) & 0 (0\%, 0\%) & 7 (42\%, 42\%) \\
Village & 3 (100\%, 100\%) & - & - \\
\bottomrule
\end{tabularx}
\end{table}

\section{Evaluation}

\subsection{Experimental Setup}
We evaluated ROCKER w.r.t. four characteristics: its runtime, RAM consumption, key extraction quality, and reduction ratio RR~\cite{pernelle2013automatic} between visited and total nodes.
\begin{equation}
RR( \alpha ) = 1 - \frac{|vnodes(\alpha)|}{2^{|\mathcal{P}|}}.
\end{equation}

Our approach was evaluated on data from twelve different datasets.
The first two datasets were chosen in order to evaluate ROCKER on an existing artificial benchmark.
Both \texttt{Restaurant 1} and \texttt{2} belong to the Ontology Alignment Evaluation Initiative (OAEI) benchmark.
We then evaluated the scalability of ROCKER on ten other datasets generated from DBpedia.
We built the datasets using the RDFSlice tool~\cite{marxsoru2014}, so that each of them contains a class with its instances and their CBD.
According to DBtrends\footnote{\url{http://dbtrends.aksw.org/}}, these classes rank among the top 20 of the most populated classes in DBpedia 3.9.
The domains vary from geography (\texttt{Village}, \texttt{ArchitecturalStructure}) to professionals (\texttt{Artist}, \texttt{SoccerPlayer}) and abstract concepts (\texttt{PersonFunction}, \texttt{CareerStation}).

\begin{table}\scriptsize
\begin{tabularx}{\columnwidth}{rXr}
\toprule
\textbf{Prefix} & \multicolumn{2}{l}{\textbf{Namespace}} \\
\midrule
\texttt{dbo:} & \multicolumn{2}{l}{http://dbpedia.org/ontology/} \\
\texttt{dbp:} & \multicolumn{2}{l}{http://dbpedia.org/property/} \\
\texttt{dcterms:} & \multicolumn{2}{l}{http://purl.org/dc/terms/} \\
\texttt{rdfs:} & \multicolumn{2}{l}{http://www.w3.org/2000/01/rdf-schema\#} \\
\texttt{geo:} & \multicolumn{2}{l}{http://www.w3.org/2003/01/geo/wgs84\_pos\#} \\
\texttt{foaf:} & \multicolumn{2}{l}{http://xmlns.com/foaf/0.1/} \\
\texttt{prov:} & \multicolumn{2}{l}{http://www.w3.org/ns/prov\#} \\
\midrule
\textbf{Size} & \textbf{Properties} & \textbf{Score} \\
\midrule
4 & [foaf:name, geo:long, dbo:location, dbp:hasPhotoCollection] & 0.99905 \\
4 & [foaf:name, geo:long, dbp:hasPhotoCollection, foaf:homepage] & 0.99905 \\
4 & [foaf:name, geo:long, dbo:elevation, dbp:hasPhotoCollection] & 0.99905 \\
4 & [foaf:name, geo:long, dbp:hasPhotoCollection, dbo:runwayLength] & 0.99905 \\
4 & [foaf:name, geo:long, dbo:openingYear, dbp:hasPhotoCollection] & 0.99905 \\
4 & [dbo:height, foaf:name, geo:long, dbp:hasPhotoCollection] & 0.99905 \\
4 & [dbo:river, foaf:name, geo:long, dbp:hasPhotoCollection] & 0.99905 \\
4 & [dbo:buildingStartYear, foaf:name, geo:long, dbp:hasPhotoCollection] & 0.99905 \\
4 & [foaf:name, geo:long, dbp:hasPhotoCollection, dbo:part] & 0.99905 \\
4 & [foaf:name, geo:long, dbp:hasPhotoCollection, dbo:primaryFuelType] & 0.99905 \\
\midrule
1 & [dbo:wikiPageID] & 0.99995 \\
1 & [rdfs:label] & 0.99995 \\
1 & [prov:wasDerivedFrom] & 0.99995 \\
1 & [dbp:hasPhotoCollection] & 0.99995 \\
1 & [foaf:isPrimaryTopicOf] & 0.99995 \\
1 & [dbo:wikiPageRevisionID] & 0.99995 \\
110 & [geo:lat, dbp:postalCode, dbp:imageFlag, dbp:northeast, \dots, dbp:arname] & 0.99997 \\
2 & [foaf:name, rdfs:comment] & 0.99958 \\
2 & [geo:long, rdfs:comment] & 0.99973 \\
2 & [geo:lat, rdfs:comment] & 0.99968 \\
2 & [rdfs:comment, dbp:name] & 0.99955 \\
2 & [dbo:wikiPageWikiLink, rdfs:comment] & 0.99914 \\
2 & [rdfs:comment, dbp:wikiPageUsesTemplate] & 0.99911 \\
2 & [dbo:isPartOf, rdfs:comment] & 0.99901 \\
2 & [dbo:wikiPageLength, rdfs:comment] & 0.99973 \\
2 & [rdfs:comment, dbo:wikiPageExternalLink] & 0.99909 \\
2 & [rdfs:comment, dcterms:subject] & 0.99902 \\
2 & [dbp:longd, rdfs:comment] & 0.99904 \\
2 & [rdfs:comment, dbp:latd] & 0.99900 \\
2 & [rdfs:comment, dbo:wikiPageOutDegree] & 0.99900 \\
\bottomrule
\end{tabularx}
\end{table}

The generation of new evaluation datasets was preferred over the use of existing datasets due to the following reasons:
\begin{enumerate}
\item Datasets from the current state-of-the-art approaches contain a maximum of $1.6M$ triples, while ours scale up to $17.1M$ triples.
\item Some of the existing datasets were not formatted properly.
\item To the best of our knowledge, no key discovery benchmark has been created to date.
\end{enumerate}

The lack of a manually-annotated gold standard for key discovery did not only affect the choice of the datasets.
This led us to adopt the number of retrieved keys and the precision to measure the key extraction quality.
In fact, while calculating the precision of a key discovery algorithm by annotating the retrieved keys is a feasible task, the set of all minimal keys needs to be known in order to compute the recall.

We compared ROCKER against two state-of-the-art approaches dubbed Linkkey~\cite{atencia2014b} and SAKey~\cite{symeonidou2014}.
While Linkkey is a tool able to retrieve keys, SAKey is more scalable and able to retrieve also $k$-almost keys (see Section~\ref{sec:score}).

ROCKER was implemented in Java as part of the link discovery framework LIMES.\footnote{\url{http://limes.sf.net}}
The datasets and the algorithm source code are also available online.\footnote{\url{http://github.com/AKSW/rocker/}}
We launched ROCKER with two different settings; the former aims at finding minimal keys ($\alpha=1$), while the latter aims at finding minimal almost-keys ($\alpha<1$). Both settings were set to use the \textit{fast search} option with $\tau=0.001$.
For the sake of simplicity, we assigned the same value to $\alpha$ (0.999) for all datasets when retrieving almost-keys.
All experiments were carried out on a 64-bit Ubuntu Linux machine with 16 GB of RAM and an octa-core 2.5 GHz CPU.

\subsection{Results}

Table~\ref{tab:results} presents the results we obtained on the twelve datasets.
Runtimes in milliseconds are reported for both tasks, i.e. ``find minimal keys'' and ``find minimal almost-keys''.
For each dataset, the size in number of triples is also shown.
As seen in table, all the computation runtimes for the artificial datasets lie within the same magnitude order of 1,000 milliseconds.
Both ROCKER runs were slower than the other approaches, however this trend has been disproved by the following results.
On the medium-sized datasets \texttt{PersonFunction}, \texttt{CareerStation} and \texttt{OrganisationMember}, our approach is the only one which completed all three tasks.
In particular, Linkkey reached the Java heap space on the first two, while SAKey did not complete on the third one.
On the seven remaining datasets whose size in NTriples format is larger than 1.5 GB, only our approach was able to finish the computation.
This fact leads to consider ROCKER as the most scalable approach for key discovery at the state of the art.

\begin{table*}[ht]\scriptsize
\centering
{\caption{Runtime results in milliseconds for ROCKER, Linkkey and SAKey on all datasets.}\label{tab:results}}
\begin{tabularx}{0.9\textwidth}{Xrrrrr}
\toprule
\textbf{Dataset} & \textbf{Triples} & \textbf{ROCKER(1.0)} & \textbf{ROCKER(0.999)} & \textbf{Linkkey} & \textbf{SAKey} \\
\midrule
OAEI 2011 Restaurant 1 & $1.1K$ & 1,880 & 2,170 & 1,698 & 1,028 \\
OAEI 2011 Restaurant 2 & $7.5K$ & 2,424 & 2,833 & 2,278 & 885 \\
DBpedia PersonFunction & $383K$ & 14,565 & 11,626 & OutOfMemory & 6,221 \\
DBpedia CareerStation & $3.0M$ & 79,964 & 118,632 & OutOfMemory & 2,199,854 \\
DBpedia OrganisationMember & $3.9M$ & 1,075,679 & 1,130,640 & 227,336 & OutOfMemory \\
DBpedia Album & $11.4M$ & 1,948,767 & 366,147 & OutOfMemory & OutOfMemory \\
DBpedia Artist & $12.0M$ & 203,764 & 168,049 & OutOfMemory & OutOfMemory \\
DBpedia Village & $12.9M$ & 4,224,338 & 18,872,456 & OutOfMemory & OutOfMemory \\
DBpedia Animal & $13.7M$ & 8,565,772 & 3,426,372 & OutOfMemory & OutOfMemory \\
DBpedia SoccerPlayer & $13.9M$ & 314,853 & 317,285 & OutOfMemory & OutOfMemory \\
DBpedia ArchitecturalStructure & $13.3M$ & 541,054 & 1,010,347 & OutOfMemory & OutOfMemory \\
DBpedia MusicalWork & $17.1M$ & 2,524,120 & 2,634,869 & OutOfMemory & OutOfMemory \\
\bottomrule
\end{tabularx}
\end{table*}

As can be read in \cite{atencia2014b}, Linkkey was evaluated on datasets smaller than all the DBpedia datasets we generated.
We thus integrated the evaluation carried out by Linkkey's authors by running the tool on our new datasets.
At the same time, the largest dataset SAKey was evaluated on is comparable with our medium-sized datasets~\cite{symeonidou2014}.
Results shown in Table~\ref{tab:results} are thus compatible with the evaluations performed by the respective state-of-the-art algorithms.

The node reduction ratio (RR) is shown in Table~\ref{tab:rratio}.
RR expresses the rate of the number of nodes that were discarded by pruning subtrees, thus avoiding to compute their scores.
The number of properties (i.e., the size of $\mathcal{P}$) and the number of visited nodes are also reported.

\begin{table*}[ht]\scriptsize
\centering
{\caption{Reduction ratios for the two settings of ROCKER on all datasets.}\label{tab:rratio}}
\begin{tabularx}{0.9\textwidth}{Xrrrrr}
\toprule
\textbf{Dataset} & \textbf{\# properties} & \textbf{vnodes(1.0)} & \textbf{vnodes(0.999)} & \textbf{RR(1.0)} & \textbf{RR(0.999)} \\
\midrule
OAEI 2011 Restaurant 1 & 4 & 6 & 6 & 60.00\% & 60.00\% \\
OAEI 2011 Restaurant 2 & 4 & 6 & 6 & 60.00\% & 60.00\% \\
DBpedia PersonFunction & 2 & 3 & 3 & 0.00\% & 0.00\% \\
DBpedia CareerStation & 3 & 4 & 4 & 42.86\% & 42.86\% \\
DBpedia OrganisationMember & 20 & 378 & 378 & 99.96\% & 99.96\% \\
DBpedia Album & 103 & 753 & 753 & $\sim$100.00\% & $\sim$100.00\% \\
DBpedia Artist & 205 & 928 & 928 & $\sim$100.00\% & $\sim$100.00\% \\
DBpedia Village & 116 & 1387 & 1700 & $\sim$100.00\% & $\sim$100.00\% \\
DBpedia Animal & 131 & 1188 & 1188 & $\sim$100.00\% & $\sim$100.00\% \\
DBpedia SoccerPlayer & 88 & 528 & 528 & $\sim$100.00\% & $\sim$100.00\% \\
DBpedia ArchitecturalStructure & 698 & 1622 & 3693 & $\sim$100.00\% & $\sim$100.00\% \\
DBpedia MusicalWork & 136 & 1201 & 1201 & $\sim$100.00\% & $\sim$100.00\% \\
\bottomrule
\end{tabularx}
\end{table*}

Table~\ref{tab:keyquality} reports the key extraction quality results.
For each dataset we show the number of outcomes and the percentage of keys and minimal keys among them (i.e., precision).
Many datasets have been omitted as no keys were found by any approach, or simply because the approach failed during the discovery (cf. Table~\ref{tab:results}).
The most interesting results appear on the two OAEI datasets, where Linkkey was not able to recognise any key.
On the other hand, SAKey was able to recognise all 3 minimal keys on \texttt{Restaurant 2}, yet it returned also 4 non-keys.
SAKey was also able to find 2 out of 3 minimal keys, 3 non-minimal keys and 3 non-keys on \texttt{Restaurant 1}.
Among the other datasets, 3 keys were found on \texttt{Village} by ROCKER only.

Namespace prefixes and the sets of almost-keys found for \texttt{DBpedia ArchitecturalStructure} resp. \texttt{DBpedia Village}, using a threshold for the discriminability score $\alpha=0.999$ are shown below.
Reported are 10 almost-keys that were found on \texttt{DBpedia ArchitecturalStructure} and the first 20 almost-keys that were found on \texttt{DBpedia Village}.
As can be seen, the algorithm found six atomic almost-keys.
After these had been removed from the maximal-priority queue, the refinement operator followed a path of 109 refinements through the same branch of the refinement tree.
Its climb ended on a node having a discriminability score greater than $\alpha$, as well as a size of 110 properties.
After removing this node and its descendants from the queue, the refinement resumed from the bottom of the graph, where ROCKER found 13 more almost-keys composed by 2 properties each.
As our algorithm found 84 almost-keys for \texttt{DBpedia Village}, the big size of most of the almost-keys may be the reason for the longest computation.

\section{Discussion}

As presented in the previous section, ROCKER improves the state of the art w.r.t. correctness and memory consumption.
Other approaches Linkkey and SAkey have shown to require much more memory than ours, as they could not return any result on bigger datasets.
In particular, the heap space of 16 GB was reached on 8 and 9 DBpedia datasets, respectively.
Unlike the other approaches, ROCKER managed to remain below the heap space by storing the hash index on disk.
In fact, in-memory-based algorithms Linkkey and SAKey were not able to handle indexes for datasets having more than 10 million triples.

\begin{figure}[h]
	\includegraphics[width=0.46\textwidth]{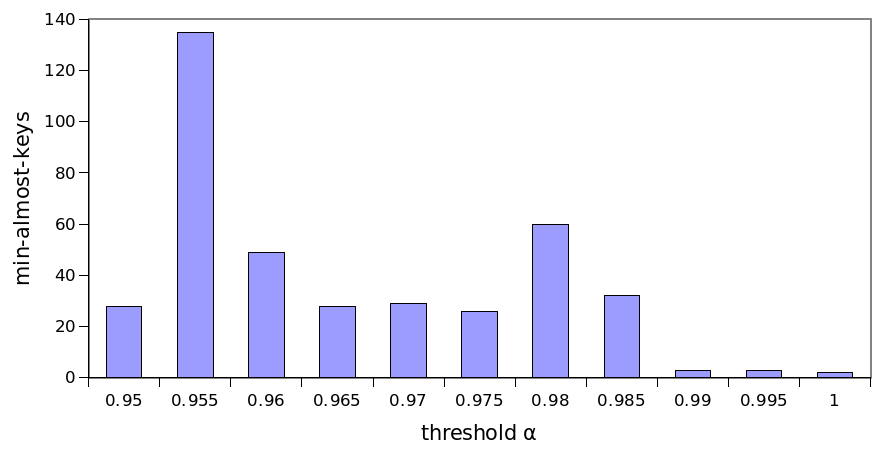}
	\vspace{-4mm}
	\caption{Number of minimal almost-keys found in function of threshold $\alpha$ for ROCKER on dataset \texttt{DBpedia Monument}.}
	\label{fig:monmakeys}
\end{figure}

Runtime results showed that SAKey is the fastest approach on small datasets, being 1.5 to 3 times faster than the others.
This could be explained by the fact that the index creation task is quicker for in-memory-based algorithms.
Moreover, the outcome analysis presented in Table~\ref{tab:keyquality} confirmed that Linkkey and SAKey found candidates that obey their respective key definitions. As mentioned before, the key definition introduced in this work is more correct.
A stricter definition leads ROCKER to a farther exploration of the knowledge graph, whereas the other approaches stop.
Thus, the runtime is affected.
Nevertheless, as can be seen, the runtime is compensated by a substantial improvement in the quality of the results.

In order to analyse how the key discovery task varies w.r.t. the threshold $\alpha$, we ran ROCKER on one chosen dataset \texttt{DBpedia Monument}.
Figure~\ref{fig:monmakeys} shows the number of minimal almost-keys found for values of $\alpha$ within the interval $[0.95, 1]$ with a step of $0.005$.
As can be seen, values are not in scale, i.e. a minimal almost-key for $\alpha_0$ does not necessarily belong to the set of minimal almost-keys for $\alpha_1 < \alpha_0$.
This is because threshold $\alpha$ can ``block'' the computation before the following refinement.
For instance, the highest value was reported for $\alpha=0.955$, where 136 minimal almost-keys were found.
Most of these keys are formed by a common root of two properties, which we call $p_1$ and $p_2$, in the form $\{p_1, p_2, p_i\}$ with $i=3,\dots,96$.
Since the discriminability score of $\{p_1, p_2\}$ is $0.953$, it is not considered as minimal almost-key for $\alpha=0.955$.
However for $\alpha=0.95$, $\{p_1, p_2\}$ will be a minimal almost-key and its descendants will not be visited, thus reducing the number of almost-keys and the runtime (see Table~\ref{fig:monruntime}).

\begin{figure}[h]
	\includegraphics[width=0.46\textwidth]{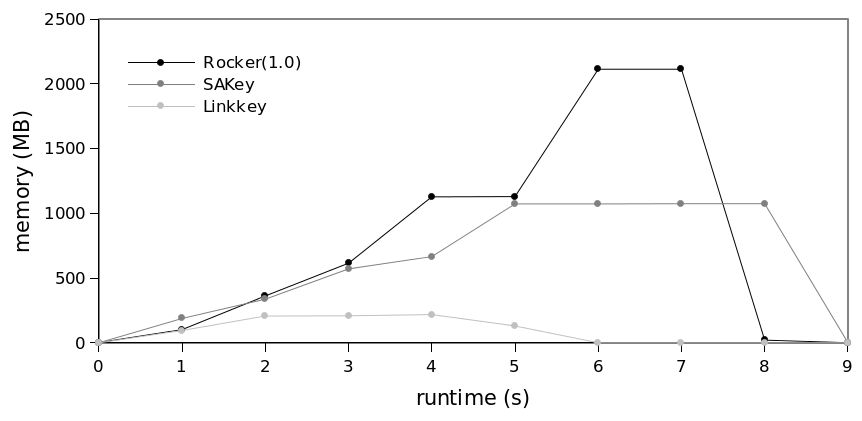}
	\caption{Linkkey showed the best runtime and RAM consumption performances on \texttt{DBpedia Monument}, confirming the results in Table~\ref{tab:results}.}
	\label{fig:monmemorywrtruntime}
\end{figure}

Table~\ref{fig:monmemory} shows the memory consumption w.r.t. $\alpha$.
For $\alpha \geq 0.99$, ROCKER required less memory ($\sim 2$ GB) than on the other experiments ($\sim 5.2$ GB), because all the almost-keys were found before visiting the remaining refinement tree.
The fact that no other almost-key exists is ensured by evaluating the score for the top element of the refinement tree, which contains all the remaining properties.
Having this a score less than $\alpha$, ROCKER ends the computation.

\begin{figure}[htb]
	\subfigure[]{
	\includegraphics[width=0.46\textwidth]{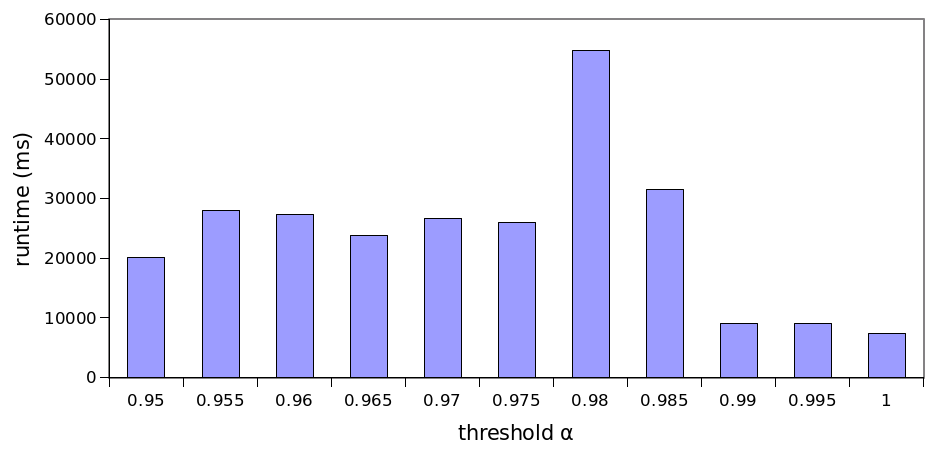}
	\label{fig:monruntime}
	}
	~
	\subfigure[]{
	\includegraphics[width=0.46\textwidth]{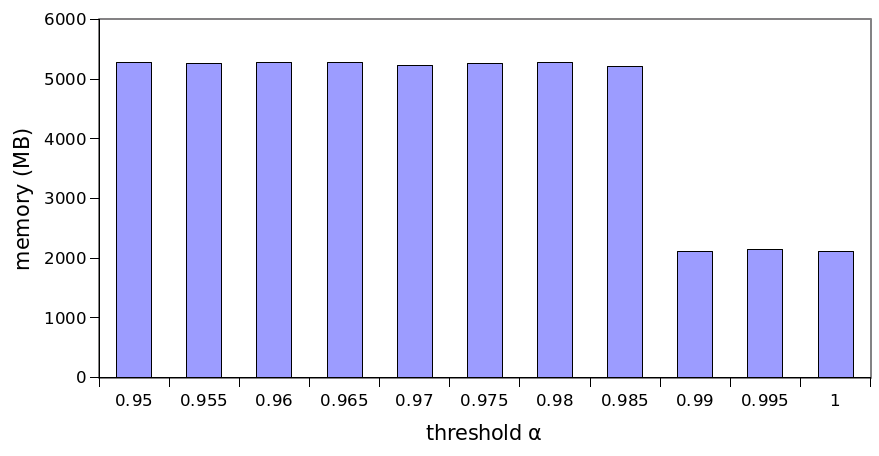}
	\label{fig:monmemory}
	}
    \caption{Run times (\ref{fig:monruntime}) and Random Access Memory consumption (\ref{fig:monmemory}) in function of threshold $\alpha$ for ROCKER on dataset \texttt{DBpedia Monument}.}
    \label{fig:monument}
\end{figure}

\section{Conclusion and Future Work}
In this paper, we presented the first refinement operator for key discovery.
We showed that the operator is non-redundant, non-complete and finite.
We implemented the operator within the ROCKER approach and showed how it can be extended to scale even on large knowledge bases.
Our evaluation of ROCKER suggests that it goes beyond the state of the art with respect to its correctness and memory efficiency, while achieving comparable runtimes.
Future directions include a study of the run times, number of keys and visited nodes w.r.t. the input threshold.
Then, we will investigate on optimization by using in-memory storage for the hash tables, in order to decrease the query runtimes.
Moreover, we will fully integrate the key discovery algorithm in LIMES and make it available in the next releases.
We will then experiment with combining key discovery with the detection of property alignments and use those alignments within the context of link discovery.

    \begin{wrapfigure}{r}{0.0\textwidth}
        \includegraphics[width=0.1\textwidth]{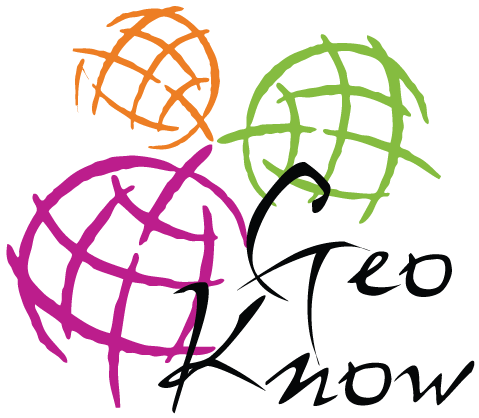}
    \end{wrapfigure}

\section{Acknowledgments}
This research is part of the GeoKnow project, funded by the European Commission with the 7th Framework Programme (Grant Agreement No. 318159).

\bibliographystyle{abbrv}
\vspace{1cm}
{\small{\bibliography{theone}}}

\balancecolumns

\end{document}